\documentclass[conference,letterpaper]{IEEEtran}
\IEEEoverridecommandlockouts
\usepackage{cite}
\usepackage{amsmath,amssymb,amsfonts}
\usepackage{algorithm}
\usepackage{algorithmic}
\usepackage{graphicx}
\usepackage{textcomp}
\usepackage{xcolor}
\usepackage{url}
\usepackage{amssymb, graphicx, amsmath, amsthm}
\usepackage{tikz}
\usepackage{multicol, blindtext}
\usepackage{subfigure}

\newtheorem{lemma}{Lemma}

\def\BibTeX{{\rm B\kern-.05em{\sc i\kern-.025em b}\kern-.08em
    T\kern-.1667em\lower.7ex\hbox{E}\kern-.125emX}}
\begin{document}
\title{The Environmental Potential of Hyper-Scale Data Centers:  
Using Locational Marginal CO$_2$ Emissions to Guide Geographical Load Shifting

}

\author{Julia Lindberg, Bernie Lesieutre and Line Roald \\
 University of Wisconsin-Madison \\
 {\underline{\{jrlindberg,  lesieutre, roald \}@wisc.edu}}}

\maketitle
\begin{abstract}

Increasing demand for computing has lead to the development of large-scale, highly optimized data centers, which represent large loads in the electric power network. Many major computing and internet companies operate multiple data centers spread geographically across the world.
Thus, these companies have a unique ability to shift computing load, and thus electric load, geographically.
This paper provides a ``bottom-up'' load shifting model which uses data centers' geographic load flexibility to lower CO$_2$ emissions. This model utilizes information about the locational marginal CO$_2$ footprint of the electricity at individual nodes, but does not require direct collaboration with the system operator. We demonstrate how to calculate marginal carbon emissions, and assess the efficacy of our approach compared to a setting where the data centers bid their flexibility into a centralized market. We find that data center load shifting can achieve substantial reductions in CO$_2$ emissions even with modest load shifting. 
\end{abstract}

\section{Introduction}

Data centers form the computing infrastructure that sustains the internet \cite{ciscointernet}, enables the revolution in artificial intelligence \cite{amodei_2019} and provides computing services and data storage for individuals and companies across the world \cite{shehabi2016united}. Between 2010 and 2018, there was an estimated 550\% increase in the number of global data center workloads and computing instances, along with a 26-fold increase in data center storage and 11-fold increase in internet protocol traffic \cite{masanet2020recalibrating}.
This increase in computing load has happened alongside a shift from smaller and medium sized data centers towards computing in large-scale facilities that are highly optimized and efficient, so-called \emph{hyper-scalar} data centers. These data centers currently consume around 1-2\% of electricity both in the United States \cite{shehabi2018data} and the world \cite{masanet2020recalibrating}. 

With the demand for cloud computing services and the number of hyper-scale data centers expected to increase, there is a growing acknowledgement that the environmental footprint of data center electricity consumption is a concern.
Hyper-scalar data centers are hosted and operated mostly by large companies like Amazon, Facebook, Google, Microsoft and Alibaba \cite{nrdc_2014, yevgeniy_analysts}.
Several of these major companies have announced policies aimed to reduce the carbon footprint of the services they provide \cite{googleenvironmentalreport, amazonenvironment}, both through improved efficiency and by investing in or contracting with renewable power generation. 
In 2019, Google matched a 100\% of their energy consumption with renewable energy purchases, and is currently working to become 24/7 carbon free through the use of so-called carbon-intelligent computing which shifts computing to less CO$_2$ intensive hours or locations \cite{googleblog}. Ongoing research aims to enable zero carbon cloud computing \cite{chien2019zero} and start-up companies are developing solutions to enable low-carbon, low-cost data center computing \cite{lancium-press}.

While the environmental footprint is important, access to reliable electricity supply is also crucial for reliable operation of the data centers. To mitigate potential bottlenecks in electricity supply, cloud computing companies take a number of steps to increase efficiency and reliability. Efficiency gains are important to curb the overall need for electricity. In the past decade, a main focus has been to increase the power usage effectiveness (PUE), defined as the ratio between the total power used by a data center to the power consumed for computation. For example, the average PUE of Google data centers during the twelve months preceding Q1 of 2020 is 1.11 PUE, down from 1.21 in 2008. 
We note that for highly efficient data centers like these, temporal load shifting achieved through, e.g., pre-cooling \cite{lukawski_tester_moore_krol_anderson_2019}, is typically not possible or effective, but instead can be achieved by, e.g., load migration, shutdown and idling of servers and storage clusters, and cooling relative to load reduction \cite{ghatikar_ganti_matson_piette_2012, li_bao_li_2015}.
Gains in efficiency can also be achieved by reducing the amount of computing required to perform a certain computing task, by, e.g., utilizing flexible, real-time load balancing algorithms for routing similar search queries to a data center that can process them more efficiently  \cite{eisenbud2016maglev}.
Furthermore, to ensure reliability, companies may maintain several copies of important data, such that some computing tasks can be performed at multiple locations in case of a data center outage.
Since large scale companies operate data centers at various locations spread across the world, they have an unprecedented opportunity for shifting load geographically among data centers, by adapting the algorithms that direct everything from search queries to large scale computing jobs. They can also shift load temporally, by deferring non-urgent computing jobs to off peak times of the day.

In this paper, we consider how hyper-scale computing companies can \emph{geographically shift computing tasks and electric load to reduce CO$_2$ emissions} from electric power generation. We will assume that (i) the data centers are large-scale, highly efficient facilities where load shifting is mainly achieved by shifting computing loads, (ii) to ensure reliability, the data centers are run in a way that enables some computing tasks to be executed at different locations, and (iii) the data centers are willing to pay a markup on their electricity represented through a price of CO$_2$ to reduce the CO$_2$ footprint of their operations. To guide data center load shifting, we propose to use \emph{locational marginal CO$_2$ emissions at individual nodes of the electric grid} which, in analogy to locational marginal prices, provide information about how an increase in load at a given location (at a given point in time) will change the overall CO$_2$ emissions of the grid.

Related work on data center load shifting investigated benefits of integrating data centers into demand response programs \cite{dcdr_survey} as well as possible demand response and pricing schemes data centers could employ \cite{liu2013data,liu2014pricing}, modelling the impact of geographical load shifting from a computational perspective, aiming to achieve reductions in cost \cite{rao_liu_xie_liu_2010, li_bao_li_2015, rao_liu_ilic_liu_2012}, considering shifts between multiple electricity markets \cite{rao_liu_xie_liu_2010} and cooperation between data centers \cite{rao_liu_ilic_liu_2012}. Others have investigated geographical redistribution to reduce the CO$_2$ footprint based on the \emph{average} amount of renewable energy in the generation mix, e.g. \cite{2015liu, zheng2020mitigating},

or the impact of data center siting on the absorption of renewable energy \cite{kim2017data,wang2015grid}. Several companies \cite{tomorrow} provide about the average CO$_2$ intensity of electricity, and \cite{zheng2020mitigating} show that the time of day matters for load shifting.  
However, none of these existing works consider that even if a region has excess renewable energy (i.e., is experiencing energy curtailment), this renewable energy may not be available in all locations in the grid. The locational aspect is captured in the concept of locational marginal CO$_2$ emissions \cite{2010Ruiz, 2011Rudkevich}, which demonstrates that nodal carbon intensity varies across the grid \cite{2010Ruiz} and is useful for guiding renewable energy investments \cite{2011Rudkevich}. However, this work did not consider data centers or load shifting.

An important aspect of our work is that we assume that \emph{the data centers and independent system operators (ISOs) do not necessarily share the same objectives or willingness to pay for CO$_2$ emission reductions}. The electricity markets clear based solely on economic cost, accounting for operational and security constraints; CO$_2$ management is not among the considerations. While this could change in the future, 
adjusting the markets is a long process. 
Therefore, we seek a more "bottom-up" method by which market participants can shift their own load to reduce the overall CO$_2$ emissions of the grid. 
To summarize, the contribution of this paper is to propose a market participant driven, bottom up approach to load shifting which relies on information regarding the locational marginal CO$_2$ emission at each network node. We use this metric to pose a user-centered optimization problem, where data centers adjust their loads the goal of reducing cost and CO$_2$ emissions. These load shifts happen outside of ISO market clearing actions, and can be realized within a grid, or between different grids.
In our case study, we compare this approach to two benchmark approaches. We first compare our proposed method to a similar bottom up approach, which uses the average CO$_2$ emissions of electricity in the grid instead of the locational marginal prices. We next compare our results with ISO-run centralized approaches where the data centers provide their shifting capacity in the energy market and/or the ISO includes CO$_2$ costs in the market clearing objective. 

The remainder of the paper is organized as follows. In Section 2 we discuss how we represent data center flexibility in shifting load. In Sections 3 and 4 we present the underlying mathematics for calculating the Locational Marginal CO$_2$ Emissions and set up the specific optimization problems used to to assess this load shifting approach. Different objective function scenarios are consider to determine the effect of including specific CO$_2$ weights in the objective function, compared to solely economic costs. In Section 5 we present the results of our analysis applied to the IEEE RTS GMLC network model. We summarize and discuss the results in Section 6.

\section{Modelling data center flexibility}
\label{sec:flexibility}

Geographical shifting of computing load requires consideration of many important aspects including latency, availability of data, reliability, and the management of computing resources to handle computational tasks over widely different scales. To ensure reliability of services like search or access to emails, companies are able to perform the necessary computing at multiple locations in case a data center experiences an outage. This reliability is on par or exceeds the reliability we are accustomed in the electric grid. 

We posit that as computation becomes more power constrained and environmentally responsible, both the cost and the carbon intensity of the local power supply may be taken into consideration in such load balancing algorithms. In this paper, we make no explicit assumptions about how this is done. We simply assume that given appropriate incentives and the right signal, highly optimized systems such as those run by technological giants like Google and Amazon can enhance existing algorithms to follow such signals within a certain set of limits.
We next develop a simplified model to represent how this computing load flexibility could translate into flexibility for spatial shifting of electric load in order to help meet environmental goals. 

\paragraph{Data center representation}
We consider data centers represented by the set $\mathcal{C}$ in the model. We assume that the initial allocation of computing loads result in the electric load $P_{d,i}$ for each of the data centers $i\in\mathcal{C}$, and denote the change in load at data center $i$ resulting from spatial shifting by $\Delta P_{d,i}$. We further introduce directed variables $s_{ij}\geq 0$ (and $s_{ji}\geq 0$) to represent the amount of load transferred from data center $i$ to data center $j$ (or $j$ to $i$). 
We further enforce that $s_{ii}=0$. This is related to the concept of virtual links described in \cite{zhang2019flexibility}.  
To represent the ability of data centers to shift load, we introduce the following constraints. 

\paragraph{Lossless load shifting}
The shifted load $\Delta P_{d,i}$ is equal to the sum of load transferred to the data center and the load transferred from the data center,

\begin{subequations}\label{eq:lossless}
\begin{align}
    \Delta P_{d,i} &= \textstyle\sum_{j\in\mathcal{C}} s_{ji} - \textstyle\sum_{k\in\mathcal{C}} s_{ik}, \quad \forall i\in{C} \label{newconst5} 
\end{align}{}
and the sum of all load shifts equal zero,
\begin{align}
    \textstyle \sum_{i\in\mathcal{C}} \Delta P_{d,i} &= 0. \label{newconst2} 
\end{align}{}
\end{subequations}

\noindent By forcing all computations to sum to zero, we ensure that all computing tasks are performed at the current time step. We do not consider temporal opportunities to delay computation for later. We further note that our formulation assumes that the same computation task will require the same amount of energy at both data centers. However, it could easily be generalized to consider the case the electricity needed to perform a certain computing task is higher (or lower) at a different location due to, e.g., access to different hardware or less immediate access to the necessary data.

\paragraph{Shifting limitations}
To represent limits on the ability of data centers to shift load, we introduce the constraints,

\begin{subequations}\label{eq:limits}
\begin{align}
    - \epsilon_i \cdot P_{d,i} &\leq \Delta P_{d,i} \leq \epsilon_i \cdot P_{d,i}, \quad &&\forall i\in\mathcal{C} \label{newconst1} \\
    0 &\leq s_{ij} \leq M_{ij} \quad &&\forall ij\in\mathcal{C}\times\mathcal{C}. \label{newconst6} 
\end{align}{}
\end{subequations}
Here, \eqref{newconst1} limits the maximum change (increase or decrease) in electric load that we can achieve at data center $i$. The maximum load shift to other locations is expressed as a fraction $\epsilon_i$ of the total data center load. Eq. \eqref{newconst6} enforces the direction of the shift $s_{ij}$ and limits it to a maximum value $M_{ij}$.  These may represent practical limits or operational choices.

\section{Data Center-Driven CO$_2$ Reduction}
\label{sec:reduction}

We assume that data centers act as price takers or retail customers in the electric markets, and can shift their load outside of the market clearing. 
We are interested in understanding how their position as large scale loads can enable the data centers to reduce the overall carbon footprint of the grid. Intuitively, it might seem rational to simply shift load to the region where the electricity has the lowest CO$_2$ intensity, e.g. the location with the highest share of renewable energy (as has been proposed in \cite{2015liu, zheng2020mitigating}).
However, this renewable energy might already be used by other loads, and other potentially CO$_2$ intensive sources of electricity might be asked to increase their generation output if the load increases further. While the average CO$_2$ intensity of the electricity consumed by the data center might be lower in such a location, the overall CO$_2$ emissions may actually increase. The question therefore becomes how to shift load in a way that replaces the use of CO$_2$ intensive generation with cleaner generation sources. We will achieve this by considering the locational marginal CO$_2$ footprint of electric loads across the grid.

Specifically, we consider a market setting where the ISO clears the market using a DC optimal power flow (OPF) formulation at short time intervals, e.g., every 5-15 min as is common in many markets across the United States. The assumption is that by observing the current market outcome the data centers try to adjust their loads to reduce the CO$_2$ emissions from generation ahead of the next market clearing\footnote{  For simplicity and proof of concept, we assume that the market clearing is frequent enough that it is reasonable to assume that the other loads and generation remains relatively constant. More realistic models will be addressed as part of future work.}.
The proposed model assumes that the data centers have  knowledge of the current locational marginal prices (LMPs) of electricity \cite{litvinov2010}, which are made publicly available in real time, and either knowledge or prediction of the marginal carbon footprint of loads at different nodes. This is currently not publicly available, but it is possible that ISOs could publish information about the marginal carbon footprint of loads in the future, or that approximate models could be developed using historical data, driven by ISO-reported current LMPs and binding constraints. 

To derive the model employed in this paper, we first revisit the standard DC OPF which is commonly used for electricity market clearing.  Next, we explain how to determine the marginal CO$_2$ footprint of electric loads, and then we describe how to use this information to reduce the overall CO$_2$ emissions of the grid.  

\subsection{DC Optimal Power Flow} \label{sec31}
We start by presenting a stylized, but representative mathematical model for the DC OPF which is commonly used for electricity market clearing \cite{christie2000transmission,litvinov2010}.

We consider an electric network with a set of $\mathcal{N}$ nodes (also commonly referred to as buses), with $|\mathcal{N}|=N$. The set of generators $\mathcal{G}$ contains a total number of $|\mathcal{G}|=N_g$ generators, and the subset of generators connected to node $i$ are denoted by $\mathcal{G}_i$.
Transmission lines connect the different nodes in the network. The set of transmission lines is denoted by $\mathcal{L}$, with element $(i,j)\in\mathcal{L}$ representing the line between node $i$ to $j$.
The set of all loads in the system is given by $\mathcal{D}$ and contains both data center loads $\mathcal{C}$ and non-data center loads $\mathcal{D}\backslash\mathcal{C}$. 
Note that there may be more than one data center load connected at each electrical node. 
The non-data center loads $j\in\mathcal{D}\backslash\mathcal{C}$ are also denoted by $P_{d,j}$, but do not have the ability to shift load, giving $\Delta P_{d,j} = 0$.
The set of loads connected at node $i$ is denoted by $\mathcal{D}_i \subset \mathcal{D}$.

The market clearing is formulated as a standard DC OPF with decision variables $x = [\theta \  P_g]$ where $P_g$ are the generation variables, $\theta$ are the voltage angles at each node and $n = N + N_g$ is the number of decision variables. The DC OPF seeks to minimize generation costs subject to demand, line flow and generation constraints, and is given by
\begin{subequations}\label{dcopf}
\begin{align}
\min_{\theta, P_g} ~~&c^T P_g \label{dcopfcost} \\
\text{s.t.} ~~& \textstyle \sum_{\ell\in\mathcal{G}_i}  \!P_{g,\ell} -\! \textstyle \sum_{\ell\in\mathcal{D}_i} P_{d,\ell} = &&
\nonumber\\
&\qquad\textstyle\sum_{j:(i,j)\in\mathcal{L}} \!\!\!\!-\beta_{ij}(\theta_i \!- \!\theta_j), &&\forall i\in\mathcal{N} \label{balance}\\
-&P^{lim}_{ij} \!\leq\! -\beta_{ij}(\theta_i \!-\!\theta_j) \!\leq\! P^{lim}_{ij}, &&\forall (i,j)\in\mathcal{L}
\label{lineineq}\\
& P^{min}_{g,i} \leq P_{g,i} \leq P^{max}_{g,i}, && \forall i\in\mathcal{G}
\label{genineq}\\
& \theta_{ref} = 0. \label{refnode} 
\end{align}
\end{subequations}
Here, the cost function minimizes the cost of generation, with $c$ representing the cost vector. Eq. \eqref{balance} is the nodal power balance constraint, while \eqref{lineineq}, \eqref{genineq} represent the transmission line and generator capacity constraints. Here, $\beta_{ij} \in \mathbb{R}$ are given susceptance values, $P_{ij}^{lim}$ is the transmission capacity (which we assume is the same in both positive and negative directions of the flow) and $P_g^{min}$ and $P_g^{max}$ are generator limits. Finally, \eqref{refnode} sets the voltage angle at the reference node to zero.

\subsection{Locational Marginal CO$_2$ Emissions}
\label{sec:LMCE}

Given a solution $x^* = [\theta^* \ P_g^*]$ to the DC OPF \eqref{dcopf}, we would like to determine the locational marginal CO$_2$ emissions (LMCE) of electricity, which we define as the change in CO$_2$ emissions for the overall system incurred by consuming an additional unit of load (i.e., 1 MWh) at a given electric node. This definition is closely related to the definition of the locational marginal price (LMP), which describes the change in overall system cost incurred by a similar change in electric consumption. While the LMPs reflect changes in the \emph{cost function} of the DC OPF and are thus easily obtained as the dual variables of the nodal power balance constraints \eqref{balance}, the situation is different when computing the sensitivity of CO$_2$ emissions to load shifts, since the CO$_2$ emissions are not reflected in the cost function. In the following, we describe how to compute the LMCE.

\paragraph{Short-hand form of the DC OPF} As a starting point, we identify that the linear optimization problem \eqref{dcopf} can be written in the short-hand form 

\begin{subequations}\label{short-dcopf}
\begin{align}
\min_x ~~&\hat{c}^T x \label{short-cost} \\
\text{s.t.}~~&Gx = h \label{short-eq}\\
&Kx\leq f \label{short-ineq}
\end{align}
\end{subequations}

\noindent where $\hat{c}$ is an extended cost vector that includes zeros for the $\theta$ variables, $G\in\mathbb{R}^{(N + 1) \times n}$ and $h\in\mathbb{R}^{N+1}$ are the parameter matrix and vector of the equality constraints \eqref{balance}, \eqref{refnode}, and $K\in\mathbb{R}^{(2N_g+2|\mathcal{L}|)\times n}$ and $f\in\mathbb{R}^{2N_g+2|\mathcal{L}|}$ are the parameter matrix and vector of the inequality constraints \eqref{lineineq}, \eqref{genineq}.

\paragraph{Optimal basis} From linear optimization theory \cite{bertsimas1997introduction}, we know that there exists at least one basic optimal solution $x^*$ to this optimization problem which has $n$ binding constraints\footnote{If several generators with the same cost are connected to one node, any combination of generation from these generators will be optimal. Here there are infinitely many optimal solutions but at least one of them has $n$ binding constraints.}.
These binding constraints include all of the equality constraints in \eqref{short-eq}, as well as a subset of the inequality constraints in \eqref{short-ineq} which are satisfied with equality at optimality. Together, this set of binding constraints form an \emph{optimal basis} $A \in \mathbb{R}^{n \times n}, b \in \mathbb{R}^n$. Without loss of generality we assume the equality constraints \eqref{short-eq} comprise the first $N$ rows of $A$, . 
Given the optimal basis $A$, we can write the system of linear equations $Ax^* = b$ which is satisfied at the optimal solution. In this model, the data center loads are fixed values that appear in the first $N$ entries in right-hand side vector $b$. We next want to consider the impact of changes to the data center load. For small changes in the load, leading to small change in $b$, we can assume that the the binding constraints at the optimal solution remain the same. In this case, the linear relation  $Ax^* = b$ can be used to calculate the LMCE sensitivities.

\paragraph{CO$_2$ sensitivity factors} Mathematically, we want to consider $A(x^* + \Delta x) = b + \Delta b$, which is equivalent to 

\begin{equation}
    A \Delta x = \Delta b. \label{sensitivity}
\end{equation}

\noindent Here, the only change in the right hand side vector is due to the change in the load, such that $\Delta d$ can be represented by

\begin{align}
    \Delta b = \left[ ~\sum_{\ell\in\mathcal{D}_1}\!\!\!\Delta P_{d,\ell}~~  \cdots ~\sum_{\ell\in\mathcal{D}_N}\!\!\!\Delta P_{d,\ell} ~~~ 0 ~ \cdots ~ 0 ~\right]^T \label{db} 
\end{align} 
We recall here that while the summation in \eqref{db} is over all loads, $\Delta P_{d,\ell}$ is non-zero only for data center loads. We also note that our formulation allows for more than one data center can be located at the same electrical node.

Given \eqref{db}, we can assess how a change in the load will change the optimal value of the decision variables $\Delta x = \left[ \Delta \theta~~ \Delta P_g\right]$.

From \eqref{sensitivity}, we obtain the linear relation
\begin{align}
    \begin{bmatrix} \Delta \theta \\ \Delta P_g   \end{bmatrix} &= A^{-1} \cdot \begin{bmatrix} \Delta P_{d} \\ 0.= \end{bmatrix}
\end{align}{}
Here, we are specifically interested in the relationship between a change in load $\Delta P_d$ and a change in the optimal generation dispatch $\Delta P_g$, given by

\begin{align}
    \Delta P_g &= B \cdot \Delta P_d, \label{newcost}
\end{align}{}
where $B$ is a matrix consisting of the last $N_g$ rows and first $N$ columns of $A^{-1}$. 

Let $g$ be a cost vector that measures the CO$_2$ emissions of each generator per MWh. Multiplying each side of $(\ref{newcost})$ on the left by $g$ gives us the following 
\begin{align}
    \Delta CO_2 = g \cdot \Delta P_g = g \cdot B \cdot \Delta P_d = \lambda_{\text{CO}_2}  \Delta P_d \label{newobj}
\end{align}{}
This provides the sensitivity of the change in CO$_2$ emissions to the change in load. 
We note that $\lambda_{\text{CO}_2}$ are local sensitivity factors that are only valid in the vicinity of the optimal solution. If the load changes $\Delta P_d$ are sufficiently large, i.e. large enough to change the set of constraints that is binding at optimum, \eqref{newobj} will only be an approximate representation.

\subsection{Optimal Data Center Load Shifting}
With the above sensitivity factors, we formulate a new optimization problem which seeks to shift data center load in a way that minimizes CO$_2$ emissions while also accounting for the cost of electricity. 
The assumption is that the load shifting will be used by the system operator in a subsequent market clearing based on the DC OPF \eqref{dcopf}, but with data center modified loads.

\paragraph{Objective function} 
Our goal is to minimize the amount of CO$_2$ emissions as expressed by \eqref{newobj} while considering the cost of electricity and the cost of other negative impacts of shifting load. To do this, we introduce a new parameter $\rho$ which represents a cost per CO$_2$ ton omitted by a generator. This parameter could either be related to a regulatory cost for CO$_2$ emissions, or represent the willingness of data centers to pay for CO$_2$ reductions. We also consider a parameter $d_{ij}$ which represents the cost of shifting load 1 MWh of load from the data center at node $i$ to the data center at node $j$, which could either be a direct monetary cost or a penalty to capture negative effects such as increased latency. With this, the objective function can be expressed as
\begin{align}
    \min_{\Delta P_d, s} \ (\rho\lambda_{\text{CO}_2}  + \lambda_{\text{LMP}} ) \Delta P_d + \!\!\! \textstyle \sum_{ij\in\mathcal{C}\times\mathcal{C}} d_{ij} s_{ij} \label{newobj2}
\end{align}{}
Here, the first term minimizes total cost, with $\rho \lambda_{\text{CO}_2}\Delta P_d$ representing the cost (or cost reductions) associated with changes in the CO$_2$ emissions due to the load shift $\Delta P_d$, and $\lambda_{\text{LMP}} \Delta P_d$ representing the change in the cost of electricity represented by the LMPs $\lambda_{\text{LMP}}$. The second term minimizes the cost of shifting load.

\paragraph{Data center flexibility} 

To ensure that the data center load shift $\Delta P_d$ respects the flexibility limits of the data centers, we include the load shift constraints \eqref{eq:lossless}, \eqref{eq:limits}.

\paragraph{Optimal Data Center Load Shifting }
With the above modelling, the optimal data center load shifting (ODC-LS) problem is given by 
\begin{align}
    \min_{\Delta P_d,s} ~~&\text{CO}_2\text{ emissions and cost \eqref{newobj2}} \label{model1} \\[-2pt]
    \text{s.t.} ~~&\text{Data center flexibility \eqref{eq:lossless}, \eqref{eq:limits}} \nonumber
    &
\end{align}

\section{Assessing the Benefit of Data Center Load Shifting} \label{sec4}

The model for optimal data center load shifting demonstrates how data centers can utilize their spatial load shifting flexibility to impact the electricity market outcomes, while reducing the overall CO$_2$ emissions from the system and/or their own electricity cost. This naturally raises the question of how the market outcome resulting from this process compares with a situation where the ISO either takes a more active role in reducing CO$_2$ emissions and/or the data centers bid their load as a service in the electricity market. A common hypothesis is that the ISO, by optimizing the use of data center flexibility to reduce cost, inherently will strive to utilize more of the cheaper (presumably renewable) generation sources and thus indirectly reduce CO$_2$ emission. 

To investigate whether this hypothesis holds true, we present three different models that combine generation scheduling, CO$_2$ emission minimization and utilization of data center load shifting in different ways. 

\subsection{Model 1: Data Center-Driven Load Shifting}
Our first model is data center driven load shifting described in Section \ref{sec:reduction}.  This model includes three steps: \\
1) The ISO solves the \eqref{dcopf} for a given load profile $P_d$.  \\
2) Provided information about the LMPs $\lambda_{LMPs}$ and the locational marginal CO$_2$ emission $\lambda_{CO_2}$, the data centers solve  \eqref{model1} to obtain the optimal load shift $\Delta P_g^*$,

    \begin{align}
        \min_{\Delta P_d} ~~&\text{CO}_2\text{ emissions and cost \eqref{newobj2}} \label{model1}\tag{\textbf{M1}} \\
        \text{s.t.} ~~&\text{Data center flexibility \eqref{eq:lossless}, \eqref{eq:limits}} \nonumber 
    \end{align}
3) The ISO solves \eqref{dcopf} with new load $P_d'=P_d+\Delta P_d^*$.

Different objective functions \eqref{newobj2} represent different preferences for cost minimization and CO$_2$ mitigation, which we will refer to as $f_{balance}$, $f_{\text{CO}_2}$ and $f_{cost}$.\\
\emph{(1) $f_{balance}$:} If we choose an intermediate value for $\rho$, we get a solution that balances the reduction in CO$_2$ emissions with the cost of electricity.\\
\emph{(2) $f_{cost}$:} If we choose to set the cost of CO$_2$ to zero, $\rho = 0$, we would obtain the solution with the lowest possible generation cost. \\
\emph{(3) $f_{\text{CO}_2}$:} If we either use a very large value for $\rho$ or ignore the cost of electricity by artificially setting $\lambda_{\text{LMP}}=0$, we obtain the solution with the lowest possible carbon emissions. \\

\subsection{Model 2: DC OPF with CO$_2$ emission cost}
Our second model is the DC OPF \eqref{dcopf} with a modified objective function. Specifically, we consider an objective function that minimizes both CO$_2$ emissions and overall generation cost, given by 
\begin{align}
  \min_{P_g, \theta} \ (\rho \cdot g^T + c^T)\cdot P_g  \label{mincarbon}
\end{align}
As above, $\rho$ is the cost associated with carbon emission, $g$ is the rate of emissions for each generator and $c$ represents the generation cost of each generator. This provides us with the following optimization model,
\begin{align}
    \min_{P_g, \theta} ~~&\text{CO}_2\text{ emissions and system cost \eqref{mincarbon}} \label{model2}\tag{\textbf{M2}} \\
    \text{s.t.} ~~&\text{DC OPF constraints \eqref{balance}- \eqref{refnode}} \nonumber 
\end{align}
As with Model 1, we consider the three different versions of the objective function, with $f_{balance}$ corresponding to intermediate values for the CO$_2$ cost $\rho$, $f_{\text{CO}_2}$ corresponding to an the objective function where we only focus on CO$_2$ emission reductions by setting $c=0$, and $f_{cost}$ corresponding to the more standard objective function where we disregard any cost of CO$_2$ by setting $\rho=0$.  Unlike Model 1, this model does not include load shifting.

\subsection{Model 3: DC OPF with load shifting}
Our third model is similar to Model 2, but assumes that the data centers bid their load flexibility into the market. Thus, the cost of data center load shifting and the data center flexibility limits must be taken into account. The nodal power balance constraint is also adapted to reflect that $\Delta P_d$ is now an optimization variable. This gives rise to the following optimization problem,
\begin{align}
    \min_{P_g, \theta, \Delta P_d} ~~&\ (\rho \cdot g^T + c^T) \cdot P_g + \textstyle\sum_{ij} d_{ij} s_{ij} \label{model3}\tag{\textbf{M3}} \\
    \text{s.t.} ~~& \textstyle\sum_{\ell\in\mathcal{G}_i}  \!P_{g,\ell} -\! \textstyle\sum_{\ell\in\mathcal{D}_i} (P_{d,\ell} + \Delta P_{d,\ell}) = 
\nonumber\\
&\qquad\quad\textstyle\sum_{j:(i,j)\in\mathcal{L}} \!\!\!\!-b_{ij}(\theta_i \!- \!\theta_j), \quad\forall i\in\mathcal{N} \nonumber\\
&\text{Data center flexibility \eqref{eq:lossless}, \eqref{eq:limits}} \nonumber\\
&\text{DC OPF inequality constraints \eqref{lineineq}-\eqref{refnode}}\nonumber
\end{align}
We consider the same three variations of the cost function $f_{balance}$, $f_{\text{CO}_2}$ and $f_{cost}$ as for Model 2. In all cases, we consider the same value for $d_{ij}$

\subsection{Relationship between models}
\label{sec:relationship}

When considering the models outlined above, some of the relationships between optimal generation cost and CO$_2$ emissions above can be deduced when the binding constraints of the original DC OPF do not change after shifting. 

Denote $p^{(i)}_{cost},p^{(i)}_{both}$ and $p^{(i)}_{carbon}$ the generation cost associated with Model $i$ using $f_{cost},f_{both}$ and $f_{carbon}$ respectively. For each model $i$ we have that $p^{(i)}_{cost} \leq p^{(i)}_{both} \leq p^{(i)}_{carbon}$ but we also observe relationships between the models. 

\begin{lemma}
$p^{(3)}_{cost} \leq p^{(1)}_{cost} \leq p^{(2)}_{cost}$
\end{lemma}
\begin{proof}
Model $2$ can be thought of as a special case of Model $1$ where $\epsilon_i = 0$ for all $i$. Model $1$ starts at the optimal value for Model $2$ then will only shift load if the cost decreases, therefore $p_{cost}^{(1)} \leq p_{cost}^{(2)}$. Similarly, any optimal solution for Model $1$ will also be feasible for Model $3$, therefore $p_{cost}^{(3)} \leq p_{cost}^{(1)}$.
\end{proof}

Similarly, if we let $E^{(i)}_{cost},E^{(i)}_{both}$ and $E^{(i)}_{carbon}$ be the CO$_2$ emissions from optimal generation profile of Model $i$ using $f_{cost}, f_{both}$ and $f_{carbon}$ respectively. As in the case of generation costs, for each model $i$ we have that $E^{(i)}_{carbon} \leq E^{(i)}_{both} \leq E^{(i)}_{cost}$. We also observe the following relationship between models.

\begin{lemma}
$E^{(3)}_{carbon} \leq E^{(2)}_{carbon}$ and $E^{(3)}_{carbon} \leq E^{(1)}_{carbon}$.
\end{lemma}
\begin{proof}
Model $2$ is a special case of Model $3$ where all $\epsilon_i = 0$, therefore, Model $3$ will only shift to a solution with fewer CO$_2$ emissions giving $E_{carbon}^{(3)} \leq E_{carbon}^{(2)}$. As above, since any optimal solution for Model $1$ is also feasible for Model $3$, $E_{carbon}^{(3)} \leq E_{carbon}^{(1)}$.
\end{proof}

\section{Case study: IEEE RTS-GLMC System}

We test the performance of the models outlined in Sections \ref{sec:reduction} and \ref{sec4} on a power network test case with a significant share of renewable generation sources.

\subsection{Test system} \label{sec51}
We test our model on the IEEE RTS-GMLC network which has has $73$ nodes, $158$ generators and $120$ lines. All  parameters of this system can be found at \cite{barrows2020the}. We designate nodes $ 14, 16, 17, 18, 19, 20, 23, 65, 66, 69, 70$ to be data centers and set the load at each of these nodes to be $400$ MW. The total load of the system is $11,681$ MW and data centers account for $4400$ MW or $37.7 \%$ of the total system load.

The generators in this network are designated as oil, coal, natural gas, hydro, nuclear, wind, storage and solar generators. We use data from the U.S. Department of Energy \cite{environmentbaseline} to get a CO$_2$ emissions factor for each type of generator. 
The CO$_2$ emissions are zero for the hydro, nuclear, wind and solar generators. For the oil, gas and coal power plants we use emission values of 0.7434, 0.9606 and 0.6042 metric tons of CO$_2$ per MWh, respectively. 
Using \cite{bataille2018carbon} as guidance, we impose a tax of $\$30$ per metric ton of CO$_2$ giving $\rho = \$30$. In addition, we set the cost of data center load shifting to $d_{ij} = \$ 0$, limit the shift between any two data centers to $M_{ij} = 400$ MW for all data center pairs $ij$ and set the maximum fraction of load that can be shifted to $\epsilon_i = 0.05$ for all data centers $i$.

This network has many nodes with multiple generators of the same type and the same cost function. In order to avoid a situation with infinitely many optimal solutions and consistently obtain a basic optimal solution, we add a small noise vector \cite{values} to the objective value to distinguish between identical generators at each node.  This creates a relative merit order for dispatch for those generators.

\subsection{Data Center-Driven Load Shifting } \label{sec51}
We first investigate the performance of Model 1 in reducing CO$_2$ emission and system cost.

\subsubsection{Load Shifting based on Locational Marginal CO$_2$ Emissions}

We start by analyzing how the data center-driven load shifting impacts the system cost and CO$_2$ emissions when we use the objective $f_{balance}$ and set $\rho=30$.

The initial DC OPF leads to a generation dispatch which costs $\$129,320$ and emits $3,977.1$ tons of CO$_2$.  There is 480.6 MW of curtailed renewable energy. If we distribute the CO$_2$ emissions equally across all loads, the data centers are responsible for $37.7\%$ or $1,498.2$ tons of CO$_2$.
Using \eqref{model1} to predict the optimal change in load leads to an increase of $+20$ MW on nodes 17, 18, 65, 69 and 70 and a decrease of $-20$ MW on nodes 16, 19, 20, 23 and 66. 
We note that each of the data centers is shifting the maximum allowable amount of 20 MW per data center, leading to a shift of 100 MW from high to low $\lambda_{\textrm{CO}_2}$ locations. Thus, the total shift cumulatively represent only $2.27 \%$ of the data center load and $0.85\%$ of the total system load.

\begin{table}
\centering
\begin{tabular}{||ccc||ccc||}
\hline
\textbf{Node} & \textbf{Shift} & \textbf{Fuel}
  & \textbf{Node} & \textbf{Shift} & \textbf{Fuel}\\
\textbf{No.} & \textbf{[MWh]} & \textbf{Type}
& \textbf{No.} & \textbf{[MWh]} & \textbf{Type} \\
\hline
9 & 3.9    & Gas & 144 & 45.8 & Solar \\
11 & -74    & Gas & 156  & 14.8 & Wind \\
41 & -31.2  & Coal & 157 & 1.6 & Wind\\
135 & 39.2   & Solar &   &  &  \\
\hline
\end{tabular}   
\vspace{2mm}
\caption{Predicted generation change after load shift.}
\label{predictedgenshift}
\end{table}

\begin{table}
\centering
\begin{tabular}{||ccc||ccc||}
\hline
\textbf{Node} & \textbf{Shift} & \textbf{Fuel}
  & \textbf{Node} & \textbf{Shift} & \textbf{Fuel}\\
\textbf{No.} & \textbf{[MWh]} & \textbf{Type}
& \textbf{No.} & \textbf{[MWh]} & \textbf{Type} \\
\hline
9 & 4.52    & Gas & 135 & 9.5 & Solar \\
11 & -6.7    & Gas & 144  & 3.4 & Solar \\
12 & -33   & Gas & 149 & 11.8 & Solar\\
13 & -33      & Gas & 150 & 11.2 & Solar\\
18 & 5.6     & Gas & 151 & 10.3 & Solar\\
41& -35.3 & Coal & 155 & -2.3 & Wind \\
74 & 4 & Nuclear & 156 & 15.7& Wind\\
127 & 32.3 & Solar & 157 & 1.9 & Wind\\
\hline
\end{tabular}   
\vspace{2mm}
\caption{Actual generation change after load shift.}
\label{demandshift2}
\end{table}

The predicted generation changes obtained from running \eqref{model1} are shown in Table~\ref{predictedgenshift}. This predicted generation shift leads to a decrease in curtailment of renewables by $101.4$ MW and a $72.3$ ton decrease in CO$_2$ emissions. When we rerun the DC OPF with the shifted load, we obtain the generation changes described in Table~\ref{demandshift2}, leading to a new generation dispatch which costs $\$126,970$ and emits $3,905.5$ tons of CO$_2$. This corresponds to a cost saving of $\$ 2350$ ($-1.82\%$), and reduction in total CO$_2$ emissions of $71.6$ tons ($-1.80\%$). Renewable energy curtailment is reduced by 97.8 MW. We observe that the CO$_2$ emission reduction and generation changes predicted by \eqref{model1} are not entirely accurate. In particular, the predicted large shifts away from the gas generator on node 11 to the solar plants on nodes 135 and 144 being smaller than expected.
 
While the CO$_2$ emission reduction percentages may seem small relative to the overall system emissions, we note this reduction was achieved by shifting only $0.85\%$ of the total system load. Furthermore, the $1.80\%$ decrease in carbon emissions was achieved while simultaneously reducing the overall cost of electricity generation. These results highlight how this approach over time could provide substantial reductions in CO$_2$ emissions without increasing cost.

\subsubsection{Impact of the choice of cost function}

We next analyze the impact of using different cost functions. Table \ref{model1comparison} (a) lists the results we get by rerunning the DC OPF after the load shifts $\Delta P_d$ obtained with each objective function $f_{\text{CO}_2}, f_{balance}$ and $f_{\text{CO}_2}$. We observe that the load shift obtained with $f_{\text{CO}_2}$ leads to the solution with the lowest CO$_2$ emissions, but the highest generation cost, while $f_{cost}$ provides the cheapest solution with the highest CO$_2$ emissions. This indicates that there is a difference between minimizing generation cost versus CO$_2$ emissions, suggesting that although renewable generation sources are typically cheaper, minimizing cost is not the same as minimizing CO$_2$ emissions. It is however worth noting that all three load shifts lead to solutions that are both cheaper and have lower CO$_2$ emissions than the original DC OPF.

\begin{table}
\centering

{\fontfamily{cmss}\selectfont \small \textbf{3 (a) Locational Marginal CO$_2$ Emissions $\lambda_{\text{CO}_2}$}}
\vspace{2mm}

\begin{tabular}{||l | c c||} 
 \hline
 \textbf{Model 1 } $\lambda_{\text{CO}_2}$ & \textbf{Cost [\$]} & \textbf{Emissions [CO$_2$ tons]} \\
 \hline
$f_{\text{CO}_2}$ & $ 127,960$ & $3,886.0~(-2.29\%)$ \\
 $f_{balance}$ & $126,970$ & $3,905.5~(-1.80\%)$  \\
 $f_{cost}$ & $126,600$ &  $3,908.5~(-1.72\%)$  \\  

\hline
\end{tabular} 
\vspace{2mm}
\label{model1comparison}

\centering

{\fontfamily{cmss}\selectfont \small \textbf{3 (b) Average CO$_2$ Emissions $\lambda_{av,\mathcal{R}}$}}
\vspace{2mm}

\begin{tabular}{||l | c c||} 
 \hline
 \textbf{Model 1 $\lambda_{av,\mathcal{R}}$} & \textbf{Cost [\$]} & \textbf{Emissions [CO$_2$ tons]} \\
 \hline
$f_{\text{CO}_2}$ & $128,680 $ & $3,980.7~(+0.09\%)$ \\
 $f_{balance}$ & $126,600$ & $3,908.5~(-1.72\%)$  \\
 $f_{cost}$ & $126,600$ &  $3,908.5~(-1.72\%)$  \\ 
\hline
\end{tabular} 
\vspace{2mm}
\caption{Cost and emissions after load shifting 
(including $\%$ change relative to original DC OPF), 
based on Model 1 and different objective functions.
}

\end{table}

\subsubsection{Comparison with Load Shifting based on Average CO$_2$ Emissions}

The locational marginal CO$_2$ emissions $\lambda_{\text{CO}_2}$ consider the marginal change in the CO$_2$ emissions that occur by increasing or decreasing load at a given node in the grid. Existing literature \cite{2015liu, zheng2020mitigating} has proposed to shift load based on the \emph{average} CO$_2$ emissions per MWh of electricity across an entire region of the grid. We now compare the performance of load shifting based on the average and locational marginal CO$_2$ shifting.

The average CO$_2$ emissions can be calculated in the following way. Given $N$ generators $P_G = [P_{G,1},\ldots, P_{G,N}]$ in one region $\mathcal{R}$ and a CO$_2$ emission value for each generator $g = [g_1,\ldots,g_N]$, the average carbon emissions $\lambda_{av, \mathcal{R}}$ is defined as
\begin{align}
    \lambda_{av, \mathcal{R}} := \textstyle \frac{g^TP_g }{ \sum_{i=1}^N P_{G,i}}
    \label{avg}
\end{align}

Using this definition, we define a new version of Model $1$ where $\lambda_{\text{CO}_2}$ is replaced with $\lambda_{av,\mathcal{R}}$ in the cost functions $f_{CO_2}$ and $f_{balance}$. We note a few qualitative differences between $\lambda_{av,\mathcal{R}}$ and $\lambda_{\text{CO}_2}$. First, with $\lambda_{\text{CO}_2}$, each node is assigned its own marginal CO$_2$ value. This value is determined under consideration of transmission grid congestion and binding generation constraints, as described in Section \ref{sec:LMCE}, and provides information about the increase in CO$_2$ emissions associated with an increase in load at this node. In comparison, with $\lambda_{av,\mathcal{R}}$, every node in region $i$ is given the \emph{same} value. Furthermore, this value provides information about the average CO$_2$ emissions associated with the current load in the region, and provides no information about how the emissions will increase or decrease if we shift additional load into the given node.  

To compute the $\lambda_{av,\mathcal{R}}$ for our test case, we use the three areas $\mathcal{R}_1, \mathcal{R}_2$ and $\mathcal{R}_3$ in the RTS-GLMC system. Region $\mathcal{R}_1$ is comprised of nodes $1-24$, $\mathcal{R}_2$ of nodes $25-48$ and  $\mathcal{R}_3$ of nodes $49-73$. Based on the intial DC OPF solution and \eqref{avg}, we obtain $\lambda_{av, \mathcal{R}_1} = 0.42$, $\lambda_{av, \mathcal{R}_2} =  0.55$ and $\lambda_{av, \mathcal{R}_3} = 0.15$. We next utilize our modified version of Model $1$ \eqref{model1} where $\lambda_{\text{CO}_2}$ is replaced with $\lambda_{av,\mathcal{R}}$, to obtain new load shifts $\Delta P_d$, and rerun the DC OPF with these load shifts.

The results for all three objective functions are given in Table~\ref{model1comparison} (b). We first observe that $f_{cost}$, which is independent of both $\lambda_{av,\mathcal{R}}$ and $\lambda_{\text{CO}_2}$ results in the same solution as in Table~\ref{model1comparison} (a). Further, in our modified model based on the average CO$_2$ emissions, $f_{balance}$ gives the same solution as $f_{cost}$, resulting in lower cost and higher CO$_2$ emissions compared with Table~\ref{model1comparison}. The most interesting result is obtained with cost function $f_{CO_2}$. When we use $f_{CO_2}$ in combination with $\lambda_{av,\mathcal{R}}$, the data center load shifting now leads to an \emph{increase} in overall CO$_2$ emissions. This happens despite a shift of $80$ MW load from regions with high average CO$_2$ emissions to regions with lower average CO$_2$ emissions. These results demonstrate that shifting data center load based on average CO$_2$ emissions may lead to unwanted effects, and highlight the value of understanding and calculating locational marginal carbon emissions.

\subsection{Comparison of Models}
Next, we compare the outcomes of the different models \eqref{model1}-\eqref{model3}, which represent varying levels of cooperation between the data centers and the ISO, in combination with the three different objective functions $f_{\text{CO}_2},~f_{balance}$ and $f_{cost}$. 
We use the same parameter values as described in Section \ref{sec51}.

The results are summarized in Table~\ref{generationcosts} (cost) and Table~\ref{co2emissionscosts} (emissions).  

\begin{table}[h!]
\centering
\begin{tabular}{||c | c c c||} 
 \hline
 &$f_{\text{CO}_2}$ & $f_{balance}$ & $f_{cost}$ \\ 
 \hline
Model $1$ & $127,960$ &$126,970$ & $126,600$  \\
  \hline 
 Model $2$ & $138,980$ & $130,860$ & $129,3200$   \\
  \hline 
  Model $3$ & $128,220$ & $108,700$ & $105,500$  \\
  \hline
\end{tabular} 

\caption{Optimal generation costs [$\$$] for the IEEE RTS GMLC System}
\label{generationcosts}

\vspace{2mm}
\centering
\begin{tabular}{||c | c c c||} 
 \hline
 &$f_{\text{CO}_2}$ & $f_{balance}$ & $f_{cost}$ \\ 
 \hline
Model $1$ & $3,886$ & $3,905.5$ & $3,908.5$  \\
  \hline 
 Model $2$ & $3,731.7$ & $ 3,795.8$ & $3,977.1$   \\
  \hline 
  Model $3$ & $3,368.7$ & $3,530.3$ & $3,707.6$  \\
  \hline
\end{tabular} 
\caption{Optimal CO$_2$ emissions [MW] for the IEEE RTS GMLC System}
\label{co2emissionscosts}
\end{table}

As expected from our analysis in Section \ref{sec:relationship}, we see that the best way to minimize cost and carbon is to use Model $3$ with $f_{cost}$ and $f_{\text{CO}_2}$ respectively. Intuitively, this is as expected since Model $3$ has the most flexibility and complete knowledge of all constraints in the system (as opposed to Model~1, which only has access to local knowledge). It is interesting to note that even when trying to decrease cost, Model $3$ gives a generation profile that emits fewer CO$_2$ emissions than Models $1$ and $2$. This suggests that the if the system operator has control over shifting loads, then not only could generation costs decrease, but CO$_2$ emissions could as well. However, this result is case specific. We have also observed instances where the data center-driven load shifting with the explicit objective of reducing CO$_2$ emissions is more effective than the indirectly aiming to reduce emission by providing more flexibility to the cost-minimizing market clearing. 

It is important to note that while the results obtained with Model~3 are more effective at reducing both cost and CO$_2$, implementing such a model would require the data centers to commit to the provision of flexibility ahead of time, and allow the ISO to make decisions that impact their operations. Finding good ways of facilitating this interaction is ongoing research. 
On the other hand, Model~1 requires limited changes to existing market structures. The main question is whether the data centers can obtain information about the locational marginal carbon footprint, either from the ISO or through independent estimates. 
We would also like to point out that Model~1 and the concept of locational marginal CO$_2$ emissions can be used to \emph{shift load between two or more different electric grids that are operated by different ISOs}. 

Finally, we want to point out that the results obtained Model~1 with cost functions $f_{cost}$ to $f_{balance}$ yield both lower cost and lower emissions relative to the results obtained from Model~2 (which represents the current market clearing and does not include load flexibility). This indicates that \emph{data centers operators could have a similar impact -- and a similar responsibility -- as electric system operators when it comes to reducing CO$_2$ emissions and maintaining access to affordable electricity}.

\section{Conclusions}

As the share of computing performed by hyper-scale data centers is increasing, the companies that operate these facilities face challenges in access to clean electricity. In this paper, we propose a bottom-up approach to load shifting where data centers utilize their ability to shift load geographically to explicitly reduce CO$_2$ emissions.
The bottom-up model relies on locational marginal CO$_2$ emissions at individual nodes of an electric power network, which provide information about \emph{the change in CO$_2$ emissions due to an increase in load at a given node}. These values change in real time in response to varying system conditions, but can be calculated based on the solution to a standard DC OPF and knowledge of generators' marginal CO$_2$ emissions.
We compare our proposed method with (1) a bottom-up load shifting model based on average CO$_2$ emissions across a grid, and (2) two centralized market clearing models, both requiring ISO participation. We find that shifting based on the locational marginal CO$_2$ emissions can achieve significant reductions in CO$_2$, while also reducing cost. Our proposed method outperforms shifting based on average CO$_2$ emissions, which in some cases lead to an increase in overall CO$_2$ emissions, but is not as effective as integrating data center flexibility into the overall market clearing.

These findings raise several directions for future work. 
It remains an open question how to compute locational marginal CO$_2$ emissions in real-time in practice, or to determine how ISOs and data centers can exchange information to achieve the best possible load shifts.

Other questions include studying the impact of CO$_2$ prices and obtaining a better understanding of cumulative CO$_2$ reductions and market impacts over time. Finally,  examining opportunities for carbon reduction via temporal flexibility could also be an effective way to reduce carbon emissions.

\bibliographystyle{ieeetr}

\bibliography{IEEEtrans}

\end{document}